\documentclass[letterpaper, 10 pt, conference]{ieeeconf}

\IEEEoverridecommandlockouts

\usepackage{graphicx} 

\usepackage{amsmath}

\usepackage{amssymb}  
\usepackage{hyperref}
\usepackage{acronym}
\usepackage{multicol}
\usepackage{color}
\usepackage{booktabs}
\usepackage{siunitx}
\usepackage{todonotes}

\newtheorem{theorem}{Theorem}

\newtheorem{definition}{Definition}

\usepackage[normalem]{ulem}

\title{\LARGE \bf Shaping oscillations via mixed feedback}

\author{Weiming Che, Fulvio Forni
	\thanks{W. Che is supported by CSC Cambridge Scholarship. W. Che and F. Forni are with the Department of Engineering, University of Cambridge, CB2 1PZ, UK {\tt\small wc289|f.forni@eng.cam.ac.uk}}}

\begin{document}
	\maketitle
	\thispagestyle{empty}
	\pagestyle{empty}

	\begin{abstract}
		We study the problem of controlling oscillations in closed loop by combining positive
		and negative feedback in a mixed configuration.  We develop a complete 
		design procedure to set the relative strength of the two feedback loops to achieve 
		steady oscillations. The proposed design takes advantage of dominance theory
		and adopts classical harmonic balance and fast/slow analysis to 
		regulate the frequency of oscillations.
		The design is illustrated on a simple two-mass system, a
		setting that reveals the potential of the approach for locomotion, 
		mimicking approaches based on central pattern generators. 
	\end{abstract}
	
	\section{Introduction}
	
	Oscillations are important system behaviors. There are rich examples in biology, like the rhythmic movements of respiration and locomotion \cite{marder2001central}, the cardiac rhythm, and several forms of biochemical oscillations \cite{Goldbeter1997}. In those examples, oscillations are robust to disturbances, yet flexible to respond to external inputs.
	The mechanisms of biological oscillations have inspired several attempts in engineering, like in robot locomotion \cite{kimura1999realization,ijspeert2007swimming}, and in neuromorphic circuits design \cite{ribar2019neuromodulation}. 
	These examples encourage the question of how to design a feedback controller to enforce
	oscillations in closed loop which are robust to perturbations yet tunable, that is, flexible enough to
	adapt their frequency and other features to the needs of specific engineering tasks, \cite{ijspeert2008central}.
	
	
	In this paper, we study the generation and tuning of oscillations through feedback. We look into a controller with 
	a mixed feedback structure, identified by two parallel feedback loops with opposite signs. 
	The reason to look into this controller is that the presence of both positive and negative feedback loops 
	is a recurrent structure in biology \cite{Smolen2001,Mitrophanov2008,tsai2008robust}, 
	with sharp examples in neuroscience \cite{Marder2014,Drion2015,Sepulchre2019}.
	Similarly in engineering, various combinations of positive and negative feedback are widespread
	in the design of  electronic oscillators  \cite{Tucker1972,Bernstein2002,Ginoux2012,Chua1987}.
	In fact, the presence of positive and negative feedback is not accidental but motivated by the 
	specific features of robustness and flexibility that this combination guarantees. In this paper
	we explore the mixed feedback control structure, proposing a design that 
	tunes the balance and strength of positive and negative feedback loops
	to achieve oscillations of desired frequencies.
	
	Finding oscillations in nonlinear systems is challenging. In contrast to the large
	body of methods to stabilize system equilibria, we have a few tools in 
	control theory to enforce and stabilize periodic trajectories. 
	The problem of designing controlled oscillations can be divided into two main components. 
	The first is to determine the existence of a stable limit cycle given the system's dynamics. 
	The principal tool here is the Poincar\'e-Bendixon theorem \cite{Hirsch1974}, which is constrained to planar systems.
	This limitation is overcome here using dominance theory \cite{forni2018differential,miranda2018analysis}, rooted in the theory of monotone 
	systems with respect to high rank cones \cite{Smith1980,Smith1986,Sanchez2009,Sanchez2010}.
	Dominance theory is able to determine when a high-dimensional
	system has a low dimensional attractor, possibly captured by planar dynamics.
	The combination of dominance theory and differential dissipativity also provides a way to characterize robustness and interconnections 
	of oscillating nonlinear systems.
	
	The second component is to shape the limit cycle, to achieve a certain
	frequency of oscillations in closed loop.
	In this paper we will take advantage of the harmonic balance method for oscillations in the quasi-harmonic regime
	\cite{Gelb1968, Mees1975, tesi1996harmonic,iwasaki2008multivariable}.
	We will also look into the literature of relay feedback systems to tackle relaxation oscillations
	\cite{aastrom1995oscillations,rand2012lecture}. 
	
	In what follows we study the oscillator design problem using the mixed feedback amplifier proposed in our previous 
	work \cite{che2020tunable}. The mixed feedback controller is tuned by two parameters:  
	the \emph{balance} $\beta$ regulates the relative strength between positive and negative feedback, and the \emph{gain} $k$ determines the collective feedback strength.
	The objective is to select balance and gain such that the closed loop oscillates at a predefined frequency. 
	We discuss why the mixed feedback structure leads 
	to robust oscillations and we use dominance theory to find the region $\mathcal{R}_{\text{osc}}$ 
	of balances $\beta$ and gains $k$ that guarantee steady oscillations. We show that the mixed feedback
	controller can achieve both quasi-harmonic oscillations and relaxation oscillations, as regulated by
	the balance parameter $\beta$. Thus, we use harmonic balance and fast/slow analysis to find the specific parameter values 
	within $\mathcal{R}_{\text{osc}}$ that guarantee the desired frequency of oscillations.
	
	The paper is organized as follows. Section \ref{sec:model} presents the mixed feedback amplifier, where 
	we also show how the combination of positive and negative feedback leads to a robust 
	destabilizing mechanism that enables oscillations. Section \ref{sec:Methods} discusses
	the main design methods, namely dominance theory, harmonic balance, and fast/slow analysis. 
	We also show how to combine these methods for design purposes.
	The discussion in Sections \ref{sec:Harmonic_Balance} and \ref{sec:fast/slow} focuses on parameter tuning
	based on harmonic balance and fast/slow analysis, with the goal of achieving
	a desired oscillation frequency in closed loop.   
	The design is illustrated on a two-mass spring-damper system in Section \ref{sec:example}, which provides a simplified, rudimentary model of locomotion.
	Taking advantage of asymmetric frictions, we show how different controlled oscillations lead to different 
	locomotion regimes. The paper is concluded in section VII.
	
	\section{The mixed feedback amplifier}\label{sec:model}
	We consider the mixed feedback amplifier proposed in \cite{che2020tunable},
	which has the structure shown in Fig. \ref{fig:sys_block}.
	The load $L(s)$ is controlled by a positive feedback channel $C_p(s)$ combined
	with a negative feedback channel $C_n(s)$. In this paper we restrict the analysis 
	to first order transfer functions
	\begin{equation}
		C_p(s) = \frac{1}{\tau_p s + 1} \qquad C_n(s) = \frac{1}{\tau_n s + 1}  \qquad \tau_n > \tau_p > 0,
	\end{equation}
	which guarantee that the phases of these transfer functions each never exceeds 90 degrees, so that the splitting 
	between negative and positive feedback channels is consistent at any frequency.
	We also assume that the load $L(s)$ has at least relative degree one,
	with poles and zeros whose real part is to the left of $- 1 / \tau_n$, 
	and that $L(0) = 1$ (normalized DC gain).
	\begin{figure}[!h]
		\centering
		\includegraphics[width=0.9\columnwidth]{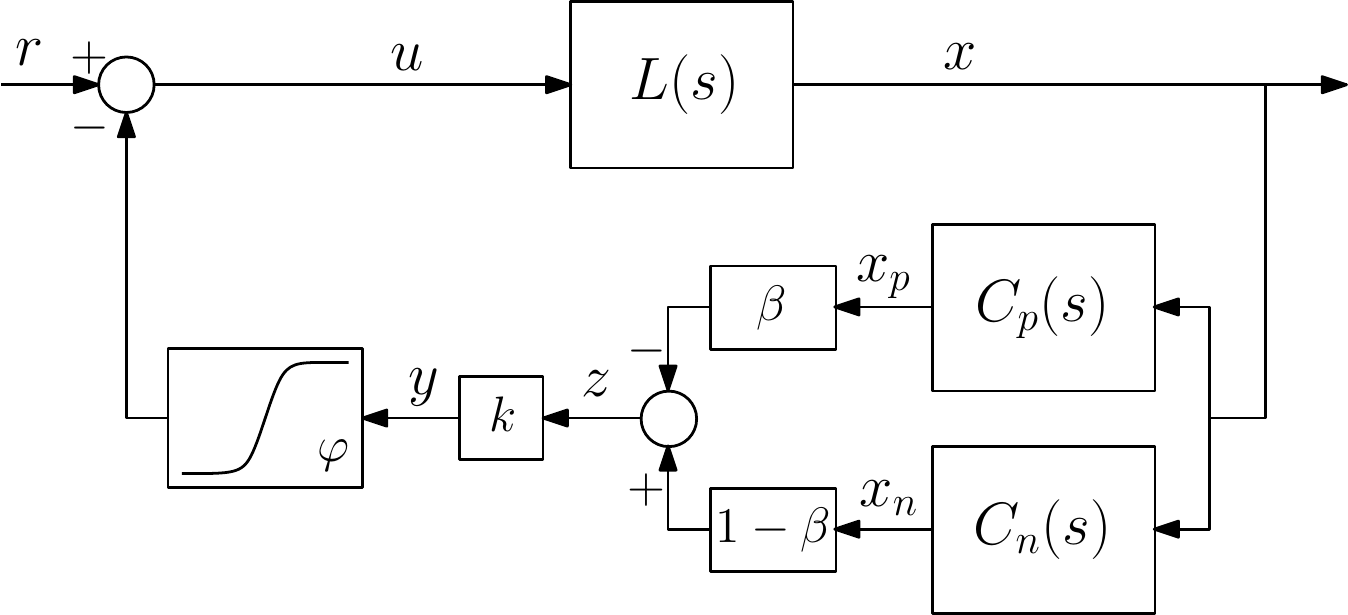}
		\caption{The Block diagram representation of the mixed feedback amplifier.}
		\label{fig:sys_block}
	\end{figure}
	
	The combination of positive and negative feedback is regulated by the parameters $k$ and $\beta$. 
	The \emph{gain} $k\geq 0$ controls the overall magnitude of the feedback and the \emph{balance} 
	$0 \leq \beta \leq 1$ controls the relative strength of two feedback channels. 
	We denote by $G(s)$ the transfer function from $u$ to $z$:	
	\begin{equation}\label{eq:G(s)}
		G(s)= C(s) L(s) = - \frac{\big(\beta(\tau_n+\tau_p)-\tau_p\big)s+2\beta-1}{(\tau_ps+1)(\tau_ns+1)}L(s)
	\end{equation}
	
	The two feedback channels are mixed and fed into a sigmoid function $\varphi$
	(bounded, differentiable, non decreasing, real function with one inflection point)
	whose slope satisfies $0 \leq \varphi'(y) \leq 1$. This monotone nonlinearity preserves the
	\emph{direction} of the control action, in the sense that $\varphi(y)y \geq 0$.
	It also guarantees \emph{boundedness} of the closed loop trajectories if
	$L(s)$, $C_p(s)$, and $C_n(s)$ have poles in the left half of the complex plane
	(by BIBO stability of $G(s)$). 
	With this representation, 
	the mixed-feedback closed loop has the structure of a Lure system, as shown in Fig. \ref{fig:Harmonic_block}. 
	
	\begin{figure}[b]
		\centering
		\includegraphics[width=0.6\columnwidth]{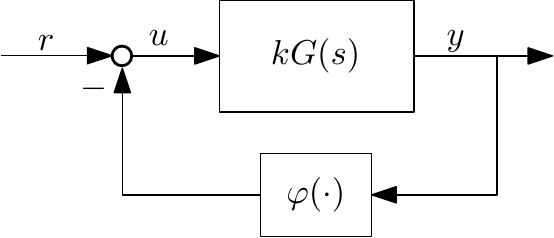}
		\caption{The Lure feedback system.}
		\label{fig:Harmonic_block}
	\end{figure}
	
	For any constant reference input $r$ (in Fig. \ref{fig:sys_block}), the closed-loop equilibria must be compatible with the equation 
	\begin{equation}
		\varphi(y) - r =\frac{y}{kG(0)}=\frac{y}{k(2\beta-1)}
	\end{equation}
	where $-k(2\beta-1)$ is the DC gain of $kG(s)$. The stability of these equilibria can be verified numerically, \cite{che2020tunable}. However, by construction, the root locus of the linearized closed-loop system 
	has the form in Fig. \ref{fig:splitting}. The red line crossing the real axis at $-\lambda$, $\lambda > \frac{1}{\tau_p}$, 
	separates the poles and the zeros of  $L(s)$ from the poles of $C(s)$. The zero $z_\beta$ of $C(s)$ can be 
	placed at any point of the positive real axis by  tuning the balance $\beta$. 
	The curves in blue represent the motion of the closed-loop poles (for increasing value of the feedback gain 
	$k$) and show how $z_\beta$ guides the loss of stability in closed loop, for sufficiently large
	gain $k$. Indeed, the mixed feedback guarantees bounded
	closed-loop trajectories, and provides a robust destabilizing mechanism, controlled by balance and gain of the feedback. 
	We use these features to control the closed loop into stable oscillations.
	
	\begin{figure}[t]
		\centering \vspace{3mm}
		\includegraphics[width=0.85\columnwidth]{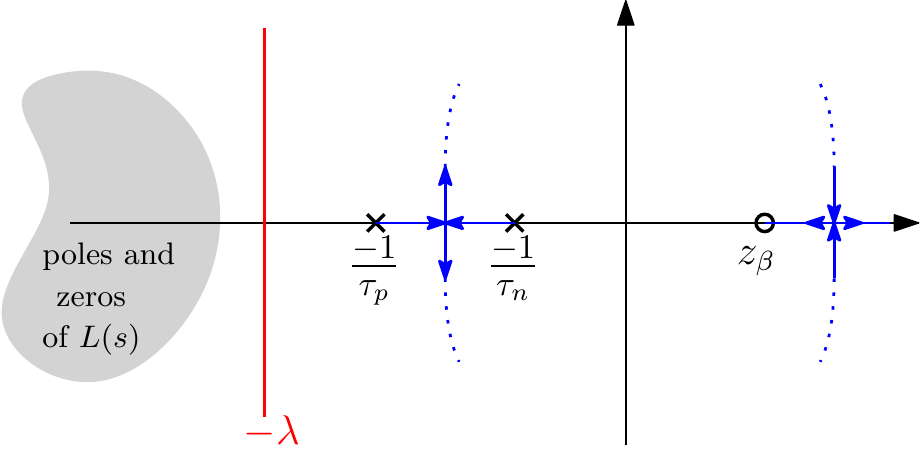}
		\caption{Root locus of $G(s)$.}
		\label{fig:splitting}
	\end{figure}
	
	\section{Determine the existence of oscillations}
	\label{sec:Methods}
	\subsection{Oscillations via $2$-dominance}
	Determining the existence of stable periodic oscillations is difficult for high dimensional systems.
	Dominance theory \cite{forni2018differential,miranda2018analysis}
	simplifies the study of periodic oscillations. 
	
	In the definition below from \cite{forni2018differential}, we use  $\partial f(x)$ to denote the Jacobian of $f$. 
	We also enforce an inertia constraint $(p,0,n-p)$ on a symmetric matrix $P$, which  means that 
	$P$ has $p$ negative eigenvalues and $n-p$ positive eigenvalues.
	\begin{definition}
		The nonlinear system $\dot{x}=f(x)$ is $p$-dominant with rate $\lambda\geq0$ if and only if there exist a symmetric matrix $P$ with inertia $(p,0,n-p)$ and $\varepsilon\geq0$ {such that the prolonged system
			\begin{equation}\label{eq:prolonged_sys}
				\begin{cases}
					\dot{x}=f(x)\\
					\delta \dot{x}=\partial f(x) \delta x
				\end{cases}\quad (x,\delta x)\in \mathbb{R}^n\times\mathbb{R}^n \ .
			\end{equation}  
			satisfies the} conic constraint
		\begin{equation}
			\label{eq:dom_inequality}
			\begin{bmatrix}
				\delta\dot{x}\\
				\delta x
			\end{bmatrix}^T \begin{bmatrix}
				0&P\\P&2\lambda P+\varepsilon I
			\end{bmatrix}\begin{bmatrix}
				\delta\dot{x}\\
				\delta x
			\end{bmatrix}\leq 0 
		\end{equation}
		along all its trajectories. The property is strict if $\varepsilon>0$.$\hfill\lrcorner$
	\end{definition} 
	
	Combining  \eqref{eq:prolonged_sys} and \eqref{eq:dom_inequality} we get the  
	Lyapunov inequality
	\begin{equation}
		(\partial f(x) + \lambda I)^T P + P (\partial f(x) + \lambda I) \leq -\varepsilon I \ .
	\end{equation}
	Given the constraint on the inertia of $P$, a necessary condition for the feasibility of this inequality is that $\partial f(x) + \lambda I$ has 
	$p$ unstable eigenvalues and $n-p$ stable eigenvalues, uniformly in $x$. 
	That is, $\partial f(x)$ must have $p$ eigenvalues to the right of $-\lambda$
	and $n-p$ to the left of $-\lambda$, a condition that is satisfied by the mixed feedback
	amplifier for $p=2$ (for small $k$ at least), as illustrated in Fig. \ref{fig:splitting}.
	
	Dominance theory provides an analytical tool to show the existence of a low dimensional attractor in a high dimensional nonlinear system, as clarified by the following theorem (\cite[Corollary 1]{forni2018differential}).
	\begin{theorem}\label{th:p-attractor} 
		For a strict $p$-dominant system with dominant rate $\lambda \ge 0$, every bounded trajectory asymptotically converges to
		\begin{itemize}
			\item a unique fixed point if $p=0$;
			\item a simple attractor if $p=2$, that is, a fixed point, a set of
			fixed points and connecting arcs, or a limit cycle.$\hfill\lrcorner$
		\end{itemize} 
	\end{theorem}
	
	Theorem \ref{th:p-attractor} clarifies that the asymptotic behavior of a $2$-dominant system 
	corresponds to the one of a planar system. This enables the use of 
	Poincar{\'e}-Bendixson-like approaches to characterize oscillations.

	We will use dominance and Theorem \eqref{th:p-attractor} to enforce oscillations in the mixed-feedback closed loop.
	In fact, the mixed-feedback closed loop is  $2$-dominant  in a specific
	range of $k$ and $\beta$, as clarified by the next theorem (adapted from \cite[Theorem 5]{che2020tunable}).	\
	\begin{theorem}\label{th:2-dominance}
		Consider a rate $\lambda$ for which the shifted
		transfer function $G(s-\lambda)$ has two unstable poles. Then, 
		for any constant $r$ and any $\beta\in[0,1]$, 
		the mixed-feedback closed loop system is 2-dominant with rate $\lambda$ 
		for any  gain $0 \leq k < \overline{k}_2 $, where 
		\begin{equation} \label{eq:k0}
			\overline{k}_2 = 
			\begin{cases}
				\infty \! & \! \mbox{if } \min\nolimits\limits_{\omega} \Re(G(j \omega-\lambda)) \!\geq\! 0 \\
				- \frac{1}{ \min\nolimits\limits_\omega \Re(G(j\omega-\lambda))} \!&\!  \mbox{otherwise. } \hfill\lrcorner 	
			\end{cases}
		\end{equation}
	\end{theorem}
	\begin{proof}
		For $k < \bar{k}_2$, $k G(j\omega - \lambda)$ lies to the right of the vertical line passing through $-1$,
		which guarantees that the closed loop is $2$-dominant by \cite[Corollary 4.5]{miranda2018analysis}.	
		Note that $\bar{k}_2$ is always greater than zero since  $G(j \omega-\lambda)$ has  finite magnitude 
		for all $\omega \in \mathbb{R} \pm\{\infty\}$.
	\end{proof}
	
	$2$-dominance and the root-locus in Fig. \ref{fig:splitting} suggest a way to 
	guarantee stable oscillations in closed loop:
	\begin{itemize}
		\item	
		we first determine the range of balances $\beta$ and gains $k$ 
		that guarantees $2$-dominance. We denote this region 
		in parameter space of $(k,\beta)$ by $\mathcal{R}_{2\text{dom}}$;
		\item
		within $\mathcal{R}_{2\text{dom}}$, stable oscillations can be precisely determined 
		by \emph{excluding} the parametric ranges where equilibria  are stable. 
		If the equilibria are unstable, Theorem \ref{th:p-attractor} guarantees that the mixed-feedback closed loop has periodic oscillations, since trajectories are bounded. We call the region of guaranteed oscillations 
		$\mathcal{R}_{\text{osc}}\subseteq\mathcal{R}_{2\text{dom}}$. 
	\end{itemize}
	For a first order load $L(s)$, an example
	is provided in Figure \ref{fig:predictions}(a). For a detailed study please refer to \cite{che2020tunable}.

	\subsection{Oscillations via harmonic balance}
	We briefly recap the classical harmonic balance method to predict oscillations, also known as describing function method \cite[Chapter 7]{khalil2002nonlinear}. The goal is to use this method to characterize the oscillations of the mixed-feedback closed loop in the quasi-harmonic regime, when possible (typically for small $\beta$).
	
	The idea is to look for periodic oscillations of the form $y(t) = E \sin(\omega t)$. We approximate the  nonlinearity $\varphi$ by its describing function, denoted by $N(E)$, which is given by the ratio between the first coefficient of the Fourier series of $\varphi(E\sin(\omega t))$ (first harmonic) and the oscillation amplitude $E$.
	An oscillation is predicted at the frequency $\omega$ that satisfies
	\begin{equation}
		kG(j\omega)=-\frac{1}{N(E)} \ .
	\end{equation}
	This corresponds to the intersection between the Nyquist plot of $kG(s)$ and the curve $-\frac{1}{N(E)}$. 
	
	There is no closed form solutions to the Fourier series of smooth odd nonlinearities $\varphi$, while they can be well-approximated by piecewise-linear function
	\begin{equation}\label{eq:piece_wise}
		\varphi_{pl}(y)=\begin{cases}
			y\quad \text{if }|y| \leq 1\\
			1\quad \text{otherwise.}
		\end{cases}
	\end{equation}
	For simplicity, in what follows we will restrict the harmonic balance analysis of the mixed feedback closed loop to $\varphi_{pl}$.
	This leads to the describing function
	\begin{equation}
		N_{pl}(E)=\begin{cases}
			1 \quad &\text{if }E\leq 1\\
			\frac{2}{\pi}[\sin^{-1}(\frac{1}{E})\!+\!\frac{1}{E}(1\!-\!\frac{1}{E^2})^{\frac{1}{2}}] \quad &\text{if }E> 1
		\end{cases}
	\end{equation}
	which is always real (as usual for odd nonlinearities) and 
	monotonically decreasing from $1$ to $0$ as $E\rightarrow\infty$.
	Indeed, our analysis will not be general for reasons of simplicity
	but can be adapted to any sigmoidal nonlinearity.

	\subsection{Relaxation oscillations via fast/slow analysis}
	\label{sec_existsosc_relaxation}
	
	Oscillations can be also predicted in the time domain. We adapt the method in \cite{aastrom1995oscillations}, 
	to determine the oscillations of the mixed-feedback closed loop
	in the relaxation regime, when possible (typically for large $\beta$).
	
	Again, we approximate $\varphi(\cdot)$ with the piecewise linear $\varphi_{pl}(\cdot)$  for simplicity. 
	For $|y| < 1$, the closed loop is a linear system.
	For $|y| \geq 1$, $|y|=1$ defines two switching planes, as shown in Fig. \ref{fig:switch_plane}.
	We assume that the traveling time between the two switching planes is negligible (fast unstable linear dynamics). 
	This corresponds to the case of a relaxation oscillation.
	In this setting,  the output $y$ jumps between positive and negative saturated value $\pm 1$,		
	producing a nearly square wave $\varphi(y)$. We can thus proceed like in \cite{aastrom1995oscillations}.
	\begin{figure}[htbp]
		\centering \vspace*{3mm}
		\includegraphics[width=0.6\columnwidth]{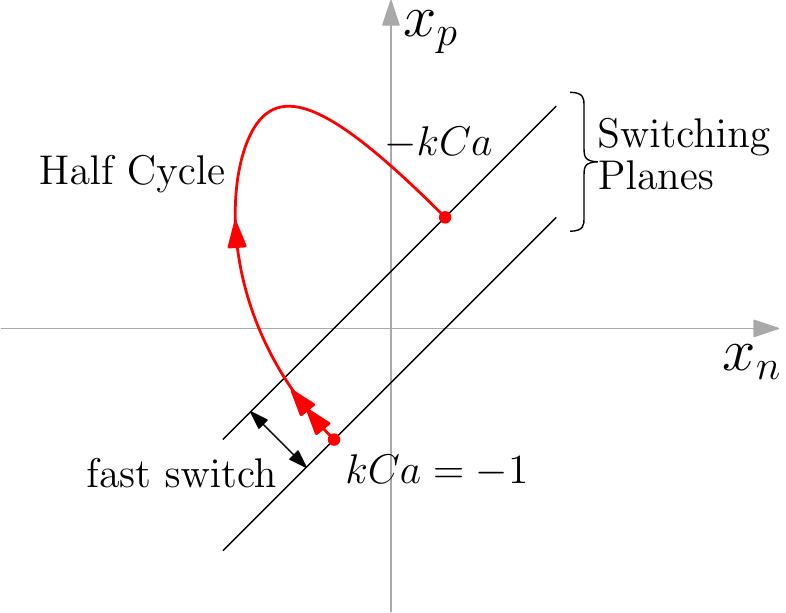} 
		\caption{An illustration of the switching planes projected on $x_p-x_n$ plane. The red curve denotes the projected state trajectory of a half cycle and `$a$' denotes the initial condition on the switching plane.} 
		\label{fig:switch_plane}
	\end{figure}
	
	Consider any minimal state-space realization $(A,B,C,0)$,
	of $G(s)$ with state given by $(\bar{x}, x_p, x_n)$, where $\bar{x}$ is the state component related to the load $L(s)$. $C$ corresponds to the matrix $\left[\begin{array}{ccccc} 0 & \dots & 0& \beta & \beta-1 \end{array}\right]$.
	The half cycle in Fig. \ref{fig:switch_plane} starts at the initial state $a$, with $y=kCa=-1$ that triggers the fast switch. Then, by symmetry, this half cycle ends at $-a$, which satisfies $-kCa=1$.
	Using the explicit solution of linear systems for constant inputs,
	\begin{equation}
		-a=e^{Ah} a-\Gamma(h) \quad 
		\Rightarrow  \quad a =(I+e^{Ah})^{-1}\Gamma(h)
	\end{equation}
	where $\Gamma(h)=\int_{0}^{h}e^{A\tau}d\tau B=A^{-1}(e^{Ah}-I)B$.
	
	Define$f(h) =Ca$. Then, the \emph{half period} of oscillation $h$
	satisfies
	\begin{equation}\label{eq:half_cycle_solution}
		kf(h)=kCa=kC(I+e^{Ah})^{-1}\Gamma(h)=-1 \ .
	\end{equation}

	\begin{figure}[!h]
		\centering
		\includegraphics[width=0.55\columnwidth]{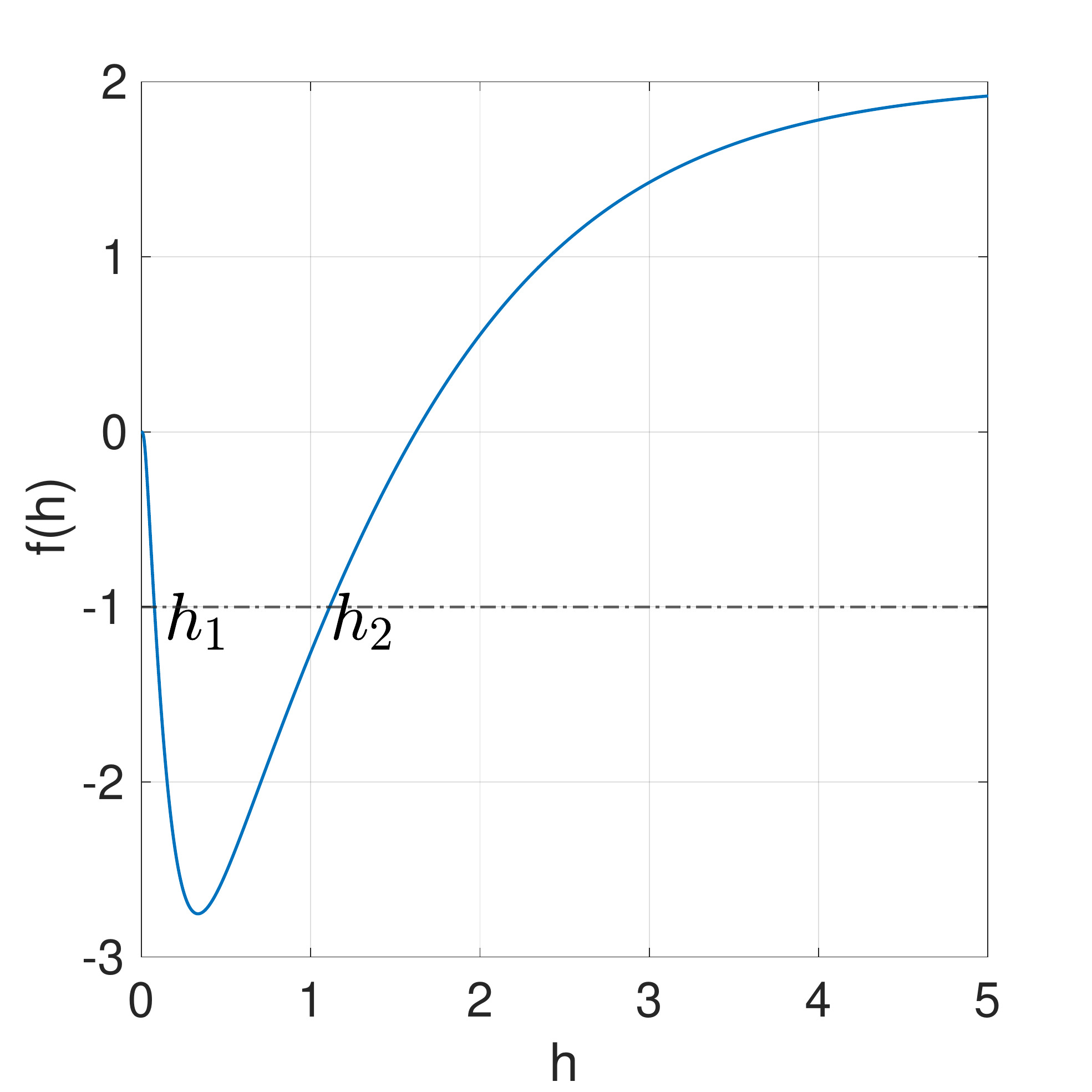}
		\caption{An example of $kf(h)$ for $k=10$, $\beta=0.4$, $\tau_\ell=0.01$, $\tau_p=0.1$, and $\tau_n=1$.
			$\tau_\ell$ is the the constant of the first order load $L(s) = \frac{1}{\tau_\ell s + 1}$.}
		\label{fig:fh}
	\end{figure}
	
	Figure \ref{fig:fh} shows a typical shape of $f(h)$ for the mixed feedback amplifier (this is obtained for a first order load $L(s) = \frac{1}{\tau_\ell s + 1}$). In general, the initial part of the curve corresponds to fast transients (the load, for example) driving the transition between switching planes. So, short time  solutions $h_1$
	should be neglected. Indeed, an oscillation with frequency $\pi/h_2$ is predicted if there is a solution that corresponds to a long interval  $h_2$ to \eqref{eq:half_cycle_solution}. This identifies the time of the half cycle illustrated in Fig. \ref{fig:switch_plane}.

	There are two situations where there is no long interval solution $h_2$ of \eqref{eq:half_cycle_solution} and the fast/slow analysis predicts no oscillations: for fixed $k> 0$, 
	\begin{enumerate}
		\item $\min\nolimits\limits_{h\geq 0} kf(h)>-1$;
		\item $\lim\limits_{h\rightarrow+\infty}kf(h)= kG(0) = -k(2\beta-1)<-1$.
	\end{enumerate} 
	1) represents the case in which $kf(h)$ has no intersections at all with the line $-1$. 2) represents
	the case in which the output does not reach the switching plane after initial transients, that is, 
	no half-cycle occurs. It follows that
	1) sets a lower bound on $k$. Furthermore, 2) constraints $k$ only when $2\beta-1>0$, that is, 
	when $\beta > \frac{1}{2}$. In that case, 2) sets an upper bound on $k$ for oscillations. 
	For illustration, these two bounds on $k$ are represented in Figure \ref{fig:predictions}(c), for
	the case of first order load. Note that  for $\beta \leq \frac{1}{2}$ there is no upper bound on $k$. 
	
	\subsection{Integration of the three methods for control design}
	
	The blue regions in Fig. \ref{fig:predictions} show that all three methods lead to similar predictions.
	These regions are obtained for the simple setting of a first order load $L(s) = \frac{1}{\tau_\ell s + 1}$, 
	for time-constants $\tau_\ell=0.01$, $\tau_p=0.1$, and $\tau_n=1$.	
	
	The difference among the three methods is that dominance analysis is not an approximated method, therefore 
	it can be used to \emph{certify} the existence of oscillations, in both harmonic and relaxation regimes. In this sense, 		dominance analysis responds to the shortcoming of harmonic balance and fast/slow analysis, due to their
	approximation natures.  At the same time, dominance analysis does not provide any information about
	the oscillation frequency of the mixed-feedback closed loop. Here harmonic balance and fast/slow analysis respond
	to the shortcoming of dominance analysis, providing guidance on the selection of $\beta$ and
	$k$ to \emph{achieve a desired  oscillation frequency} in closed loop.
	
	The idea is thus to integrate these methods to achieve a reliable control design of oscillators. We use dominance theory to determine the parameter range that guarantees oscillations. Within this range, we use either harmonic balance or fast/slow analysis for controlling the oscillation frequency.

	\begin{figure*}[htbp]
		\centering \vspace{3mm}
		\includegraphics[width=1\textwidth]{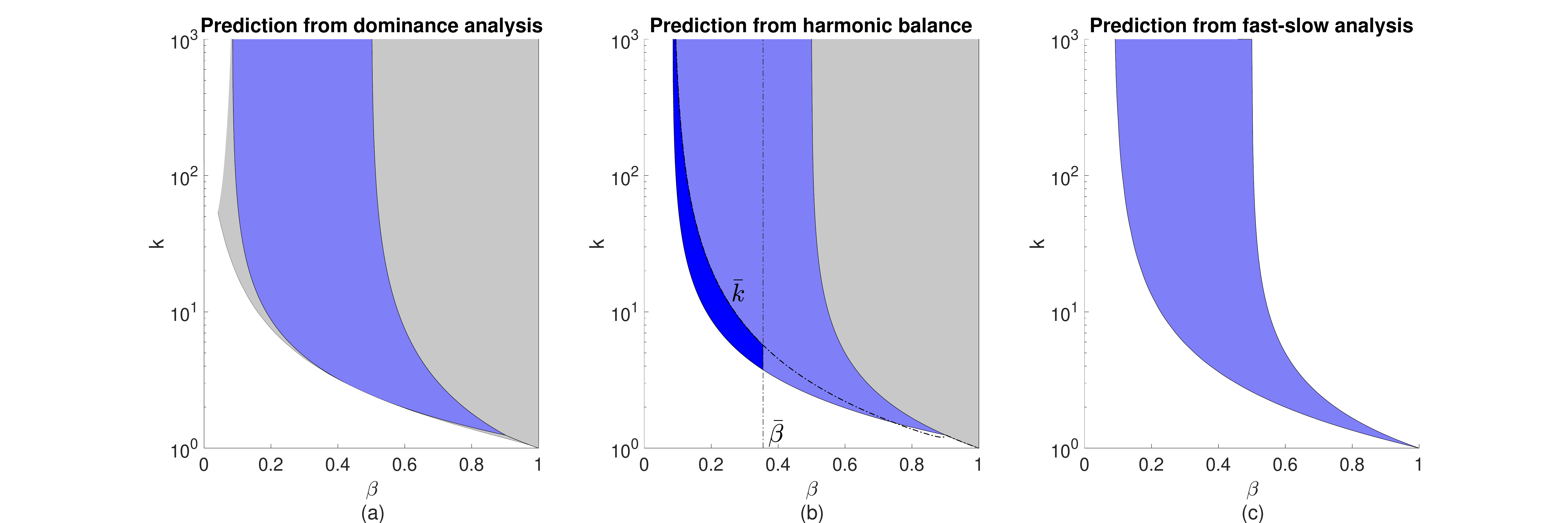}
		\caption{Predictions of the mixed feedback amplifier with $L(s) = \frac{1}{\tau_\ell s + 1}$ and $\tau_\ell=0.01$, $\tau_p=0.1$,  $\tau_n=1$, and $r=0$. \textbf{(a) Dominance analysis.} Grey region - 2-dominant region,
			$\mathcal{R}_{2\text{dom}}$ (stable and oscillatory regimes). Blue region - oscillations region $\mathcal{R}_{\text{osc}}$ (unstable equilibrium). \textbf{(b) Harmonic balance.} Gray region - multiple intersection between Nyquist plot and $N(E)$. Blue region - one intersection between Nyquist plot and $N(E)$. Dark blue region: accurate prediction of oscillations ($G(s)$ is low pass). \textbf{(c) Fast-slow analysis.} Blue region - oscillations with estimated period $\pi  / h_2$.}
		\label{fig:predictions}
	\end{figure*}	
	
	\section{Frequency shaping via harmonic balance}
	\label{sec:Harmonic_Balance}
	\subsection{Parameter range for accurate prediction}
	The harmonic balance method assumes that the linear subsystem $G(s)$ is a low pass filter
	such that higher order harmonic signals in $\varphi(y)$ are suppressed. 
	This means that to achieve reliable tuning of the mixed-feedback amplifier using the harmonic balance method, we need to identify gain and balance regimes, $k$ and $\beta$, that lead to accurate predictions.
	Consider a desired frequency of oscillations $\omega_r$:
	\begin{itemize}
		\item[i)] 
		we need $\angle G(j\omega_r)=-180^\circ$. This is achieved by selecting $\beta$,
		which moves the zero $z_\beta=\frac{1-2\beta}{\beta(\tau_p+\tau_n)-\tau_p}$, shaping
		the phase of $G$. For a low-pass $G$, the magnitude of $z_\beta$ should be at least larger than the magnitude of the smallest pole of $G(s)$. This sets an upper bound 
		\begin{equation}\label{eq:beta_bar}
			\beta < \bar{\beta}=\inf_{\beta\in[0,1]}\left|\frac{1-2\beta}{\beta(\tau_p+\tau_n)-\tau_p}\right|\geq\frac{1}{\tau_n} \ ;
		\end{equation} 
		\item[ii)] the gain $k$ modulates the magnitude of $G(s)$.
		To filter out the higher harmonics, 
		$G(s)$ must have gain less than $1$ for frequencies $n\omega_r$,  $n \geq 2$. 
		For each $\beta<\bar{\beta}$, this enforces an upper limit  $\bar{k}(\beta)$ on $k$ given by 
		\begin{equation}
			\label{eq:bark}
			\bar{k}(\beta)=\frac{1}{|G(2j\omega_r)|}.	
		\end{equation} 
	\end{itemize}
	
	As an example, $\bar{\beta}$ and $\bar{k}$ are illustrated in Fig. \ref{fig:predictions}(b) for the case of a first order load. The dark blue region guarantees accurate predictions. For $k$ and $\beta$ not in the dark blue region, higher order harmonics are not attenuated; prediction accuracy degrades and oscillations start to transform into relaxation oscillations. For $k$ and $\beta$ not in the dark blue region, prediction accuracy degrades since higher order harmonics are not attenuated. Oscillations start to transform into relaxation oscillations. Thus for design using harmonic balance method, we search for $k<\bar{k}$ and $\beta<\bar{\beta}$.

	\subsection{Design procedure}
	Within the parametric range discussed in the last section, we use harmonic balance to solve the following 
	problem: \emph{find $(k,\beta) \in \mathcal{R}_{\text{osc}}$, 
		$k < \bar{k}$, and $\beta < \bar{\beta}$ such that the mixed-feedback closed loop 
		oscillates with desired frequency $\omega_r$}. 
	
	The key is to  find a $\beta<\bar{\beta}$ such that $\angle(G(j\omega))=-180^\circ$. Denote the phase contribution of
	the numerator and denominator of $-C(j\omega_r)$ respectively by $\theta_n$ and $\theta_d$. The phase of $G(s)$ at $\omega_r$ is
	\begin{equation}
		\label{eq:phase_balance}
		\angle G(j\omega_r) =\angle(-1)+\theta_d+\theta_n + \angle L(j\omega_r)=180^\circ+2m\pi
	\end{equation} 
	Note that both $\beta$ and $k$ only affect $\theta_n$ (the phase of the denominator is 
	affected by the selection of the time constants $\tau_n$ and $\tau_p$). Thus, for any given $\omega_r$,	
	\begin{itemize}
		\item find $\beta$ such that $\theta_n = -\theta_d - \angle L(j\omega_r) - 2m\pi$.
		\item if $\beta<\bar{\beta}$, select a suitable $k<\bar{k}$
		\item If $\beta \geq\bar{\beta}$ , repeat the design for different 
		time constants $\tau_p$ and $\tau_n$.
	\end{itemize} 
	
	We close this section by clarifying the relationship between $\beta$ and $\theta_n$. Note that 
	$\theta_n=\arctan\big(\frac{(\beta(\tau_p\!+\!\tau_n)-\tau_p)\omega_r}{2\beta-1}\big)$, whose range depends on the sign of $\beta(\tau_p+\tau_n)-\tau_p$ and $2\beta-1$. Therefore, $\theta_n$ falls into three different ranges for $\beta\in[0,1]$:
	\begin{table}[!h]
		\centering
		\begin{tabular}{c|ccc} \toprule
			{$\beta\in$} & $[0,\frac{\tau_p}{\tau_p+\tau_n})$ & $[\frac{\tau_p}{\tau_p+\tau_n},0.5)$ & $[0.5,1]$ \\ \midrule
			{$\theta_n\in$} & $[180^\circ,270^\circ]$ & $[90^\circ,180^\circ]$ & $[0^\circ,90^\circ]$\\ \bottomrule
		\end{tabular}
	\end{table}
	
	The discussion above suggests that the desired oscillation frequency $\omega_r$ must decrease as $\beta$ increases. This is because the range of $\theta_n$ gets smaller, therefore it is harder to balance the
	phase contribution of $\theta_d + \angle L(j\omega_r)$ for large frequencies. 
	
	\section{Frequency shaping via fast/slow method}
	\label{sec:fast/slow}
	\subsection{Parametric range for accurate prediction}
	The first step for fast/slow method is also to quantify a sub parametric range $(k,\beta)$  of $\mathcal{R}_{\text{osc}}$ which gives accurate predicted frequencies. The main source of the approximation error is the presence of a non-negligible switching transient, that is, the time the system needs to move between
	switching planes in Fig. \ref{fig:switch_plane}. This time is reduced for large $\beta$, which increases 
	positive feedback and pushes the control signal towards saturation. Likewise, the distance $d$ among
	the switching planes is also reduced by larger feedback gains $k$, since
	$
	d=\frac{2}{k\sqrt{2\beta^2-2\beta+1}}.
	$
	
	Unlike the harmonic balance method, there is no quantitative criteria, such as loop shape and gain, to set the range of $(k,\beta)$. We thus restrict the use of the fast/slow method to those parameters for which harmonic balance is unreliable ($k>\bar{k}$, and $\beta > \bar{\beta}$) with lower bounds $\bar{k}$ and $\bar{\beta}$ defined in Section \ref{sec:Harmonic_Balance}.
	
	\subsection{Design procedure}
	The design problem for the fast/slow method is as follows: \emph{find $(k,\beta) \in \mathcal{R}_{\text{osc}}$, 
		$k > \bar{k}$ and $\beta > \bar{\beta}$ such that the mixed-feedback closed loop oscillates with desired frequency $\omega_r$}.
	
	The desired frequency $\omega_r$ defines the half cycle period $h_r=\pi/\omega_r$. 
	For this period, $(k,\beta)$ are found by solving $kf(h_r)=-1$. Remember that 
	\begin{equation}\label{eq:half_cycle_design}
		kf(h_r)=kC(I+e^{Ah_r})^{-1}A^{-1}(e^{Ah_r}-I)B=-1 \, ,
	\end{equation}
	which shows that $k$ and $\beta$ appear only in 
	$kC = k\left[\begin{array}{ccccc} 0 & \dots & 0& \beta & \beta-1 \end{array}\right]$.
	This strongly simplifies the search for feasible pairs $(k,\beta)$.	
	
	In contrast to the harmonic balance method where each desired frequency $\omega_r$ 
	corresponds to a specific balance $\beta$,
	the fast/slow method is more flexible, possibly offering several solutions of $(k,\beta)$ for the same
	desired  frequency $\omega_r$.

	\section{Example: controlled oscillations of a two-mass system}
	\label{sec:example}
	As an illustration, we present an example of two-mass system in Fig. \ref{fig:mass-sys}. 
	With asymmetric frictions at contact with ground, oscillations of the two masses lead to a 
	positive average displacement of the center of mass of the system. In this setting, the 
	two-mass system is a basic model for studying locomotion. 
	
	\begin{figure}[!h]
		\centering
		\includegraphics[width=0.2\textwidth]{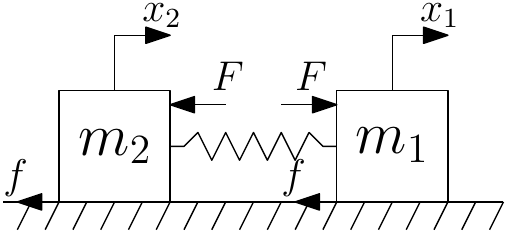}
		\caption{The double mass system.}
		\label{fig:mass-sys}
	\end{figure}
	
	The dynamics of the two-mass system \cite{majewski2011oscillatory} satisfies
	\begin{equation}
		\begin{cases}
			\ddot{x}_1=-k_m(x_1-x_2)-d_m(\dot{x}_1-\dot{x}_2)+ \gamma F-f(\dot{x}_1)\\
			\ddot{x}_2=k_m(x_1-x_2)+d_m(\dot{x}_1-\dot{x}_2)-\gamma F-f(\dot{x}_2)
		\end{cases}
	\end{equation}
	where $k_m=100$ and $d_m=10$ are normalized elastic and damping coefficients, respectively;
	$F$ is the (internal) force produced by the actuator, scaled by a factor $\gamma=100$ for simplicity; and 
	$f$ models asymmetric friction forces.
	
	The objective is to design a feedback controller that drives the two-mass system into oscillations,
	by acting on the force $F$. To keep the design within the linear setting of this paper we 
	design our controller by neglecting the asymmetric friction forces, taking  $f=0$. 
	Define $w=x_1-x_2$. For $f =0$ we have
	$
	\ddot{w}=-2k_mw-2d_m\dot{w}+2\gamma F 
	$,
	which gives the load  transfer function
	\begin{equation}\label{eq:2mass_tf}
		L(s)=\frac{200 }{s^2+20s+200} \ .
	\end{equation}
	$L(s)$ has poles at $-10\pm10j$. In the notation of Fig \ref{fig:sys_block}, $w$ is the output load $L(s)$
	block, which is used by the mixed feedback channel to generate the control signal $F = \varphi(y)$ (simulations use $\varphi=\tanh$).
	We consider $r=0$ and we set the positive and negative feedback time constants as $\tau_p=1$ and $\tau_n=10$. With this load and time constants, we get 
	\begin{equation}
		G(s)=\frac{-200\Big(\big(1.1\beta-1\big)s+2\beta-1\Big)}{(s^2+20s+200)(s+1)(10s+1)}
	\end{equation}
	As an illustration, we consider the designing of two oscillation frequencies, $\omega_r=1$ rad/s and $\omega_r=0.1$ rad/s.
	
	We use harmonic balance method for $\omega_r=1$ rad/s.
	From \eqref{eq:beta_bar}, the balance upper bound  is $\bar{\beta}=0.3548$, under which the linear system is  low pass. For $\omega_r=1$, \eqref{eq:phase_balance} leads to $\beta=0.1538$.
	Using \eqref{eq:bark} we get the gain  upper bound $\bar{k}=28.9494$.
	By dominance analysis, oscillations exist for $k>14.5217$ for $\beta=0.1538$. We thus choose $k=20$.

	By using the harmonic balance procedure for  $\omega_r=0.1$ rad/s,  \eqref{eq:phase_balance} leads to the selection 
	$\beta=0.8226$ which is much greater than $\bar{\beta}$, meaning that the oscillation at such low frequency is of the relaxation type. Hence we switch to the fast/slow method. From Section \ref{sec:fast/slow}, we
	take $h_r=\pi/\omega_r\approx31.4$ s. 
	Then, $\beta_r=0.5$ and $k=24$ is a solution to \eqref{eq:half_cycle_design}
	These parameters are compatible with dominance, which guarantees oscillations for 
	any $k>3.1584$ given $\beta$.
	
	Fig. \ref{fig:eg_simu1} presents the simulation results for $F = \varphi(y)$, where $y$ is generated
	by the mixed feedback channels in  Fig. \ref{fig:sys_block}, for $f = 0$ (no asymmetric frictions).
	The frequencies of the simulations agree with the specifications.
	\begin{figure}[htbp]
		\centering
		\includegraphics[width=0.93\columnwidth]{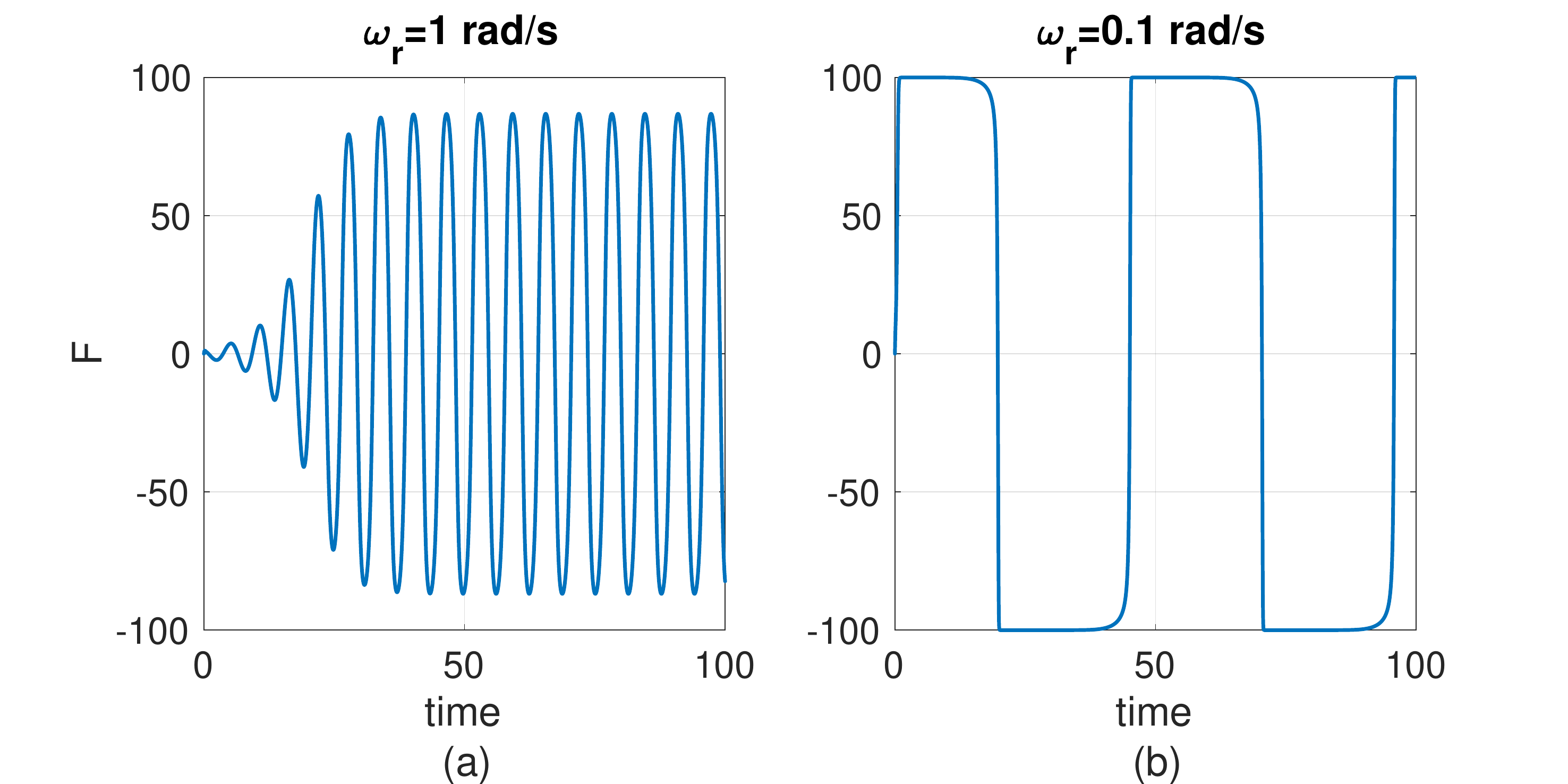} 
		\caption{Closed-loop controlled oscillations of the two-mass system. Left: desired frequency $\omega_r=1$ rad/s, achieved frequency $0.9906$ rad/s. Right: desired frequency $\omega_r=0.1$ rad/s, achieved frequency $0.1238$ rad/s.} 
		\label{fig:eg_simu1}
	\end{figure}
	
	We now reintroduce the asymmetric friction forces $f$. A complete analysis of the robustness of the
	oscillations is beyond the scope of this paper. We just emphasize that the mixed feedback induces 
	a hyperbolic instability of the closed loop equilibria, 
	which in turn guarantees robustness of oscillations for small perturbations. This is illustrated
	through simulations. For  $i\in\{1,2\}$
	we take the forward and backward friction forces as:
	\begin{equation}
		f(\dot{x}_i)
		= \begin{cases}
			5 \quad &\text{if } \dot{x}_i > 0\\
			0 \quad &\text{if } \dot{x}_i = 0\\
			-20	\quad &\text{if } \dot{x}_i<0 .
		\end{cases}
	\end{equation}
	
	\begin{figure}[htbp]
		\centering 
		\includegraphics[width=0.93\columnwidth]{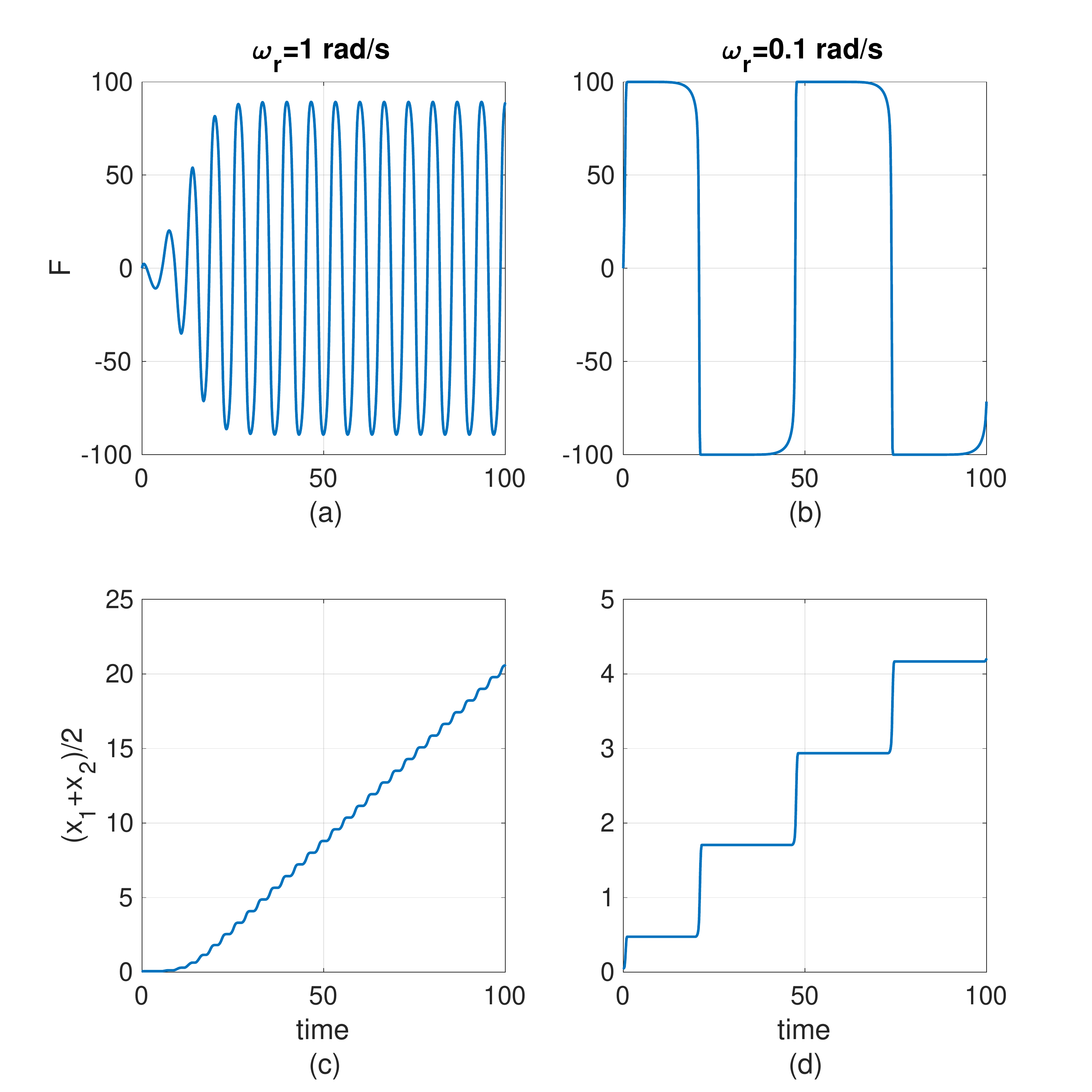}  
		\caption{Closed-loop oscillations of the two-mass system with frictions 
			Left: desired frequency $\omega=1$ rad/s, achieved frequency $0.9374$ rad/s. Right: desired frequency $\omega_r=0.1$ rad/s, achieved frequency $0.1185$ rad/s.}
		\label{fig:eg_simu2_with_f}
	\end{figure}
	
	The closed loop maintains its oscillation patterns with mild frequency changes,  as shown in Fig. \ref{fig:eg_simu2_with_f}, (a) and (b). Fig. \ref{fig:eg_simu2_with_f}, (c) and (d) illustrate the forward
	motion of the system.

	\section{Conclusions}
	We study the problem of controlling oscillations in closed loop by combining positive
	and negative feedback in a mixed configuration. This is illustrated by developing
	a complete design, using dominance theory to set balance $\beta$ and gain $k$ 
	to achieve reliable oscillations, and harmonic balance and fast/slow analysis
	to regulate those oscillations towards a desired frequency.
	The design is illustrated on a simple two-mass system, where the mixed
	feedback regulates oscillations to achieve locomotion, emulating approaches
	based on central pattern generators.
	
	In contrast to classical entrainment due to a driving external source, 
	generating endogenous oscillations \emph{through feedback} opens the way
	to questions of sensitivity/robustness of the oscillations to interconnections. Understanding how the frequency of oscillations is shaped by the interaction with an external system is a relevant question to enable adaptive control schemes in applications.

	\bibliographystyle{ieeetr}

\end{document}